\documentclass[12pt]{article}
\usepackage{amsmath,amsthm,amssymb,color,verbatim,graphicx,fullpage,url,cite}
\usepackage{tcolorbox}
\newcommand{\remove}[1]{}
\newtheorem{theorem}{Theorem}[section]
\newtheorem{claim}[theorem]{Claim}
\newtheorem{lemma}[theorem]{Lemma}

\newtheorem{definition}[theorem]{Definition}
\newtheorem{corollary}[theorem]{Corollary}
\newtheorem{conjecture}[theorem]{Conjecture}

\newtheorem{open}[theorem]{Open problem}

\newcommand{\eps}{\varepsilon}
\newcommand{\R}{\mathbb{R}}
\newcommand{\Z}{\mathbb{Z}}
\newcommand{\E}{\mathbb{E}}
\newcommand{\ip}[1]{\langle{#1}\rangle}

\newcommand{\sign}{\text{sign}}

\newcommand{\A}[1]{\mathcal{A}_{#1}}
\newcommand{\eqdef}{:=}
\newcommand{\infer}{\text{infer}}

\newcommand{\restate}[2]{\medskip \noindent {\bf #1 (restated).} {\sl #2}\\}

\title{Near-optimal linear decision trees for k-SUM and related problems}

\author{
Daniel M. Kane\thanks{Department of Computer Science and Engineering/Department of Mathematics, University of California, San Diego. {\tt dakane@ucsd.edu} Supported by NSF CAREER Award ID 1553288 and a Sloan fellowship.}
\and Shachar Lovett\thanks{Department of Computer Science and Engineering, University of California, San Diego. {\tt slovett@cs.ucsd.edu.} Research supported by NSF CAREER award 1350481, CCF award 1614023 and a Sloan fellowship.}
\and Shay Moran\thanks{Department of Computer Science and Engineering, University of California, San Diego, 
Simons Institute for the Theory of Computing, Berkeley, and Max Planck Institute for Informatics, Saarbr\"{u}cken, Germany. {\tt  shaymoran1@gmail.com.}}
}

\date{}

\begin{document}
\maketitle

\begin{abstract}
We construct near optimal linear decision trees for a variety of decision problems in combinatorics and discrete geometry.
For example, for any constant $k$, we construct linear decision trees that solve the $k$-SUM problem on $n$ elements
using $O(n \log^2 n)$ linear queries.
Moreover, the queries we use are comparison queries, which compare the sums of two $k$-subsets;
when viewed as linear queries, comparison queries are $2k$-sparse and have only $\{-1,0,1\}$ coefficients.
We give similar constructions for sorting sumsets $A+B$ and for solving the SUBSET-SUM problem, both with optimal number of queries,
up to poly-logarithmic terms.

Our constructions are based on the notion of ``inference dimension",
recently introduced by the authors in the context of active classification with comparison queries.
This can be viewed as another contribution to the fruitful link between machine learning and discrete geometry,
which goes back to the discovery of the VC dimension.
%
\end{abstract}

\section{Introduction}


This paper studies the linear decision tree complexity of several combinatorial problems,
such as $k$-SUM, SUBSET-SUM, KNAPSACK, sorting sumsets, and more.
A common feature these problems share is that they are all instances of
the following fundamental problem in computational geometry.

\paragraph{The point-location problem.}
Let $H \subset \R^n$ be a finite set.
Consider the problem in which given $x\in \R^n$ as an input,
the goal is to compute the function
\[\A{H}(x) \eqdef \bigl(\sign(\ip{x,h}): h \in H\bigr) \in \{-,0,+\}^H,\]
where $\sign:\R \to \{-,0,+\}$ is the sign function and $\ip{\cdot,\cdot}$ is the standard inner product in~$\R^n$.

In discrete geometry this is known as the \emph{point-location in an hyperplane-arrangement problem},
in which each $h\in H$ is identified with the hyperplane orthogonal to $h$,
and $\A{H}(x)$ corresponds to the cell in the partition induced by the hyperplanes in $H$
to which the input point $x$ belongs.

A dual formulation of this problem has been considered in learning theory, specifically within the context of active learning:
here, each $h\in H$ is thought of as a point, $x$ is thought of as the learned half-space,
and computing $\A{H}(x)$ corresponds to learning how each point $h\in H$
is classified by $x$.
In this work it will often be more intuitive to consider this dual formulation.
See Figure~\ref{fig:hyp} for a planar illustration of both interpretations.

\begin{figure}
\begin{center}
\includegraphics[width=.75\textwidth]{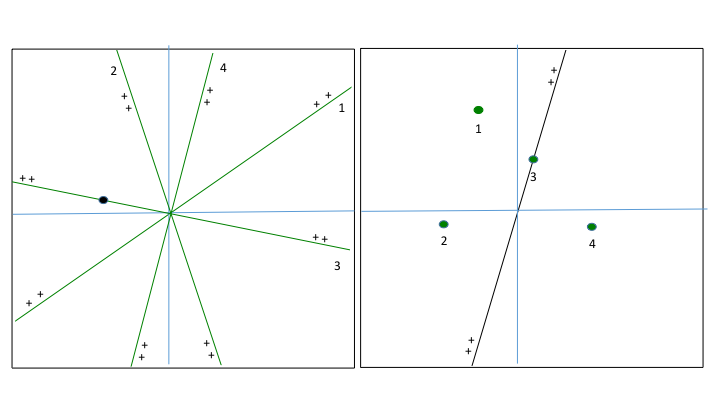}
\end{center}
\caption{ Primal and dual forms of the point-location problem.
$H$ is the green lines/points, and $x$ is the black point/line, and $\A{H}(x)= (+,+,0,-)$ }
\label{fig:hyp}
\end{figure}

\paragraph{Linear decision tree.}
A \emph{linear decision tree} for the point-location problem $\A{H}$ is an adaptive deterministic algorithm $T$.
The set $H \subset \R^n$ is known in advance, and the input is $x \in \R^n$. The algorithm does not have direct access to $x$. Instead,
at each iteration the algorithm chooses some $h \in \R^n$ and queries ``$\sign(\ip{h,x})=?$'' (note that $h$ is not necessarily in $H$). At the end, the algorithm should be
able to compute $\A{H}(x)$ correctly. The \emph{query complexity} is the maximum over $x$ of the number of queries performed. Equivalently, such
an algorithm can be described by a ternary decision tree which computes the sign of a linear query at each inner node. A query is $s$-sparse
if it involves at most $s$ nonzero coefficients. A linear decision tree is $s$-sparse if all its queries are $s$-sparse.

\paragraph{Comparison decision tree.}
A comparison decision tree for the point-location problem $\A{H}$
is a special type of a linear decision tree, where the only queries used
are either of the form $\sign(\ip{h,x})$ for $h \in H$ (\emph{label queries}),
or $\sign(\ip{h'-h'',x})$ for $h',h'' \in H$ (\emph{comparison queries}).
Note that $\ip{h'-h'',x} \geq 0$ if and only if $\ip{h',x}\geq \ip{h'',x}$,
which is why we call these comparison queries.
In the dual version (in which we view $H$ as a set of points),
comparison queries have a natural geometric interpretation:
assuming that $\sign(\ip{h',x})=\sign(\ip{h'',x})$, a comparison query $\ip{h'-h'',x}$,
corresponds to querying which one of $h',h''\in H$ is further from the hyperplane
defined by $x$.
Observe that if all elements $h \in H$ are $s$-sparse then a
comparison decision tree is $2s$-sparse.

\subsection{Results}
Our main result is a method that produces near optimal decision trees for many natural and well studied combinatorial instances for the point-location
problems by using comparison decision trees. We first describe a few concrete instances, and then the general framework.

\subsubsection{$k$-SUM}
In the $k$-SUM problem an input array $x\in\R^n$ of $n$ numbers is given,
and the goal is to decide whether the sum of $k$ distinct numbers is $0$.
This problem (in particular $3$-SUM) has been extensively studied since the 1990s,
as it embeds into many problems in computational geometry,
see for example~\cite{gajentaan1995class}. More recently, it has also been studied in the context of
fine-grained complexity, see for example the survey~\cite{vassilevska2015hardness}.

The $k$-SUM problem corresponds to the following point-location problem.
Let $H\subseteq\{0,1\}^n$ denote all vectors of hamming weight $k$.
Thus, $x\in\R^n$ contains $k$ numbers whose sum is $0$
if and only if $\A{H}(x)$ contains at least one $0$ entry.

In this context, comparison decision trees allow for two types of linear queries:
label queries of the form
``$\sum_{i \in I} x_i \ge 0 ?$'' where $I \subset [n]$ has size $|I|=k$,
and comparison queries of the form ``$\sum_{i \in I} x_i \ge \sum_{j \in J} x_j ?$'' where $I,J \subset [n]$ have size $|I|=|J|=k$.

\begin{theorem}\label{thm:ksum}
The $k$-SUM problem on $n$ elements can be computed by a comparison decision tree of depth $O(kn \log^2 n)$.
In particular, all the queries are $2k$-sparse and have only $\{-1,0,1\}$ coefficients.
\end{theorem}

This improves a series of works. There is a simple algorithm
based on hashing that solves $k$-SUM in time $O(n^{\lceil k/2 \rceil})$. It can
be transformed to a linear decision tree with the same number of queries, which in our language are all label queries.
Erickson~\cite{erickson1995lower} showed
that $\Omega(n^{\lceil k/2 \rceil})$
queries are indeed necessary to solve $k$-SUM if only label queries are allowed (or more generally, if only $k$-sparse linear queries are allowed).
Ailon and Chazelle~\cite{ailon2005lower} extended the lower bound, and showed that if the linear queries have sparsity less than $2k$,
than a super-linear lower bound of $n^{1+\Omega(1)}$ holds for the number of queries (note that indeed the near-linear comparison decision tree
given by Theorem~\ref{thm:ksum} is $2k$-sparse).

In a breakthrough work, Gr{\o}nlund and Pettie~\cite{gronlund2014threesomes} were the first to break the $n^{\lceil k/2 \rceil}$ bound.
They constructed a randomized $(2k-2)$-linear decision tree for $k$-SUM which makes $O(n^{k/2} \sqrt{\log n})$ queries.
This was improved to $O(n^{k/2})$ by Gold and Sharir~\cite{gold2015improved}.

In the general linear decision tree model, without any sparsity assumptions, a series of works in discrete geometry have designed
linear decision trees for the general point-location problem. In the context of $k$-SUM, the best result is
of Ezra and Sharir~\cite{ezra2016decision}, who constructed a linear decision tree of depth $O(n^2 \log^2 n)$ for any constant $k$.
This improves on previous results of Meyer auf der Heide~\cite{meyer1984polynomial}, Meiser~\cite{meiser1993point} and
Cardinal et al.~\cite{cardinal2015solving}.

\subsubsection{Sorting $A+B$}
Let $A,B \subset \R$ be sets of size $|A|=|B|=n$. Their sumset, denoted by $A+B$ is the set $\{a+b: a\in A,b \in B\}$.
Consider the goal of sorting $A+B$ while minimizing the number of comparisons
(here, by comparisons we mean the usual notion in sorting, that is comparing two elements of $A+B$).
While it is possible that $|A+B|=n^2$, it is well known that the number of possible orderings of $A+B$ is only $n^{O(n)}$~\cite{fredman1976good}.
Thus, from an information theoretic perspective it is conceivable that $A+B$ can be sorted using
only $O(n \log n)$ comparisons.
However, Fredman~\cite{fredman1976good} gave a tight bound of $\Theta(n^2)$ on the number
of comparisons needed to sort $A+B$.

It is natural to ask whether enabling the algorithm more access to the data
in the form of simple local queries can achieve $o(n^2)$ query-complexity.
We show that if the algorithm can use \emph{differences-comparisons}
than an almost optimal query-complexity of $O(n\log^2n)$ suffices to sort $A+B$.
A differences-comparison on an array $[x_1,\ldots,x_n]$ is a query of the form
\[\text{``}x_i-x_j\geq x_k-x_l\;?\text{''};\]
in words:
``is $x_i$ greater than $x_j$ more than $x_k$ is greater than $x_l$?''.

The problem of sorting $A+B$ corresponds to the
following point-location problem. Let $A=\{a_1,\ldots,a_n\}, B=\{b_1,\ldots,b_n\}$ and identify $x \in \R^{2n}$
with $x=(a_1,\ldots,a_n,b_1,\ldots,b_n)$. Let $H \subset \{-1,0,1\}^{2n}$ consist of vectors with exactly one $1$ and one $-1$
in the first $n$ elements, and exactly one $1$ and one $-1$ in the last $n$ elements. Then computing $\A{H}(x)$ corresponds to answering
all queries of the form ``$a_i+b_j \ge a_k + b_l ?$" for all $i,j,k,l \in [n]$, which amounts to sorting $A+B$.
In this context, the two types of queries used by comparison decision trees
are comparison queries in $A+B$, namely ``$a_i+b_j \ge a_k+b_l ?$" where $i,j,k,l \in [n]$
(which correspond to the label queries in the point location problem),
and differences-comparison queries in $A+B$,
namely ``$a_i+b_j-a_{i'}-b_{j'} \ge a_k+b_l-a_{k'}-b_{l'} ?$" where $i,j,k,l,i',j',k',l' \in [n]$
(which correspond to comparison queries in the point location problem).
%

\begin{theorem}\label{thm:sortAB}
Given $A,B \subset \R$ of size $|A|=|B|=n$, their sumset $A+B$ can be sorted by a comparison decision tree of depth $O(n \log^2 n)$.
In particular, all queries are $8$-sparse with $\{-1,0,1\}$ coefficients.
\end{theorem}

The problem of sorting sumsets has been considered by Fredman~\cite{fredman1976good}, who
showed that if only comparison queries are allowed, then $\Theta(n^2)$ queries are sufficient and necessary
to sort $A+B$. Gr{\o}nlund and Pettie~\cite{gronlund2014threesomes} use it in their work,
and specifically ask for a better linear decision tree for sorting sumsets.

\subsubsection{NP-hard problems}
Several NP-hard problems can be phrased as point-location problems. For example,
the \emph{SUBSET-SUM} problem is to decide, given a set $A$ of $n$ real numbers,
whether there exists a subset of $A$ whose sum is $0$.
The \emph{KNAPSACK} problem is to decide whether there exists a subset of $A$ whose sum is $1$.
We focus here on
SUBSET-SUM for concreteness.

The SUBSET-SUM problem corresponds to the
following point-location problem. Let $A=\{a_1,\ldots,a_n\}$ and take $x=(a_1,\ldots,a_n) \in \R^n$.
Let $H = \{0,1\}^{n}\backslash \{0^n\}$. Then $A$ has a subset whose sum is $0$ if and only if $\A{H}(x)$ contains at least one $0$.

In this context, comparison decision trees have two types of queries:
label queries of the form ``$\sum_{i \in A'} a_i \ge 0 ?$" for some $A' \subseteq A$,
and comparison queries of the form ``$\sum_{i \in A'} a_i \ge \sum_{i \in A''} a_i ?$" for some $A',A'' \subseteq A$.

\begin{theorem}\label{thm:subsetsum}
The SUBSET-SUM problem can be solved using a comparison decision tree of depth $O(n^2 \log n)$,
where $n$ is the size of the input-set.
In particular, all the queries are linear with $\{-1,0,1\}$ coefficients.
\end{theorem}

Note that the bound is tight up to the log factor:
indeed, in the corresponding point-location problem, $H=\{0,1\}^n$, and thus $\{\A{H}(x) : x\in\R^n\}$
corresponds to the family of thresholds function on the boolean cube. It is well known that the number
of such functions is $2^{\Theta(n^2)}$~\cite{DBLP:conf/ifip/GotoT62}, and thus any decision tree (even one that uses arbitrary queries,
each with a constant number of possible answers)
that computes $\A{H}(x)$ must use at least $\Omega(n^2)$ queries.

The surprising fact that SUBSET-SUM, an NP-hard problem, has a polynomial time algorithm in a nonuniform model
(namely, linear decision trees)
was first discovered by Meyer auf der Heide~\cite{meyer1984polynomial},
answering an open problem posed by Dobkin and Lipton~\cite{dobkin1974some} and Yao~\cite{yao1981parallel}.
It originally required $O(n^4 \log n)$ linear queries. It was generalized by Meiser~\cite{meiser1993point} to the general
point-location problem, and later improved by Cardinal~\cite{cardinal2015solving} and Ezra and Sharir~\cite{ezra2016decision}.
This last work, although it does not address the SUBSET-SUM directly, seems to improves the number of queries to $O(n^3 \log^2 n)$.
Observe that our construction gives a near-optimal number of linear queries, namely $O(n^2 \log n)$. Moreover,
the queries are simple, in the sense that they involve only $\{-1,0,1\}$ coefficients,
and natural from a a computational perspective as they only compare the sums of subsets.
This is unlike the previous works mentioned, which requires arbitrary coefficients due to the geometric nature of their techniques.

\subsubsection{Other applications}
Our framework (see Corollary~\ref{cor:general}) is pretty generic, and as such gives near optimal linear decision
trees for a host of problems considered in the literature. For example, the following problems were considered in~\cite{gronlund2014threesomes}.
We discuss each one briefly, and refer the interested reader to~\cite{gronlund2014threesomes} for a deeper discussion.

\paragraph{$k$-LDT.} Given a fixed linear equation $\phi(x_1,\ldots,x_k) = \alpha_0 + \sum_{i=1}^k \alpha_i x_i$ and a set $A \subset \R$ of size $|A|=n$,
the goal is to decide if there exist distinct $a_1,\ldots,a_k \in A$ such that $\phi(a_1,\ldots,a_k)=0$. This problem is a variant of the $k$-SUM
problem, and can be embedded as a point-location problem in $\R^{nk+1}$ as follows. Let
$x=(1,\alpha_1 a_1, \ldots, \alpha_1 a_n, \ldots, \alpha_k a_1, \ldots, \alpha_k a_n)$ and $H \subset \{-1,0,1\}^{nk+1}$
consists of $h$ which have a ``$-1$'' in their first coordinate, a single ``$+1$'' in each of the $k$ blocks of size $n$, and $0$ elsewhere.
Corollary~\ref{cor:general} implies a comparison decision tree with $O(kn \log^2 n)$ queries which are $(2k+2)$-sparse and with $\{-1,0,1\}$ coefficients.
For constant $k$ this gives $O(n \log^2 n)$, which improves upon the previous best bound of $O(n^2 \log^2 n)$ of~\cite{ezra2016decision}.

\paragraph{Zero triangles.} Let $G=(V,E)$ be a graph on $|V|=n$ vertices and $|E|=m$ edges, which is known in advance
(it is not part of the input).
The inputs are edge weights $x:E \to \R$.
The goal is to decide if there is a triangle in $G$ whose sum is zero. This problem clearly embeds as a point-location problem in $\R^m$.
Corollary~\ref{cor:general} gives a comparison decision tree which solves this problem with $O(m \log^2 m)$ queries. All the queries
are $6$-sparse and have $\{-1,0,1\}$ coefficients. This improves upon the previous bound of $O(m^{5/4})$ of~\cite{gronlund2014threesomes}.

\subsection{General framework}

Our results are based on the notion of ``inference dimension", which was recently introduced by the authors~\cite{kane2017active}
in the context of active learning.

\begin{definition}[Inference]
Let $S \subset \R^n$ and $h,x \in \R^n$. We say that \emph{$S$ infers $h$ at $x$} if ``$\sign(\ip{h,x})$''
is determined by the answers to the label and comparison queries on $S$. That is, if we set
\[
P_S(x) \eqdef \{x' \in \R^n: \A{S \cup (S-S)}(x')=\A{S \cup (S-S)}(x)\}
\]
then $\sign(\ip{x',h})=\sign(\ip{x,h})$ for all $x' \in P_S(x)$.
We further define the \emph{inference set} of $S$ at $x$ to be
\[
\infer(S,x) \eqdef \{h \in \R^n: S \text{ infers } h \text{ at } x\}.
\]
For each $h \in \infer(S,x)$, we refer to $\sign(\ip{h,x})$ as the \emph{inferred value} of $h$ at $x$.
\end{definition}

An equivalent geometric condition to ``$S$ infers $h$ at $x$" is that the hyperplane defined by $h$ is either disjoint from $P_S(x)$ or contains $P_S(x)$.

For example, if $h_1,h_2$ are such that $\sign(\ip{h_1,x})=\sign(\ip{h_2,x})=0$,
and $h$ is in the linear space spanned by $h_1,h_2$ then $\sign(\ip{h,x})=0$ and so $\{h_1,h_2\}$
infer $h$ at $x$. Similarly, if $\sign(\ip{h_1,x})=\sign(\ip{h_2-h_1,x})=+1$,
and $h$ is in the cone spanned by $h_1,h_2-h_1$ (i.e.\ $h=\alpha h_1 + \beta (h_2-h_1)$ for $\alpha,\beta>0$)
then $\sign(\ip{h,x})=+1$ and so $\{h_1,h_2\}$ infer $h$ at $x$.

\begin{definition}[Inference dimension]
Let $H \subset \R^n$. The inference dimension of $H$ is the minimal $d \ge 1$ for which the following holds. For any subset $S \subset H$
of size $|S|\geq d$, and for any $x \in \R^n$, there exists $h \in S$ such that $S \setminus \{h\}$ infers $h$ at $x$.
\end{definition}

We refer the reader to~\cite{kane2017active}
for some simple examples and further discussion regarding the inference dimension.

The first step in the proof of Theorem~\ref{thm:ksum}, Theorem~\ref{thm:sortAB} and Theorem~\ref{thm:subsetsum},
is to show that the sets $H$ in the corresponding point location problems are of low inference dimension.
The following general theorem provides a uniform treatment for this.

For $h \in \Z^n$ defines it $\ell_1$ norm as $\|h\|_1 = \sum_{i=1}^n |h_i|$.
\begin{theorem}
\label{thm:infdim}
The inference dimension of $H = \{ h\in \Z^n :\|h\|_1 \le w\}$ is $d=O(n \log w)$.
\end{theorem}

Next, we show that sets of low inference dimension have efficient comparison decision trees.
As a first step, we show this for \emph{zero-error randomized comparison decision trees}.
A zero-error randomized comparison decision tree
is a distribution over (deterministic) comparison decision trees $T$,
each solves $\A{H}(x)$ correctly for all inputs. The \emph{expected query complexity}
is the maximum over $x$, of the expected number of queries performed by $T(x)$ to compute $\A{H}(x)$.

\begin{theorem}
\label{thm:zeroerror}
Let $H \subset \R^n$ be a finite set with inference dimension $d$. Then there exists a zero-error randomized comparison decision tree which computes $\A{H}$,
whose expected query complexity is ${O\bigl((d +n\log d) \log |H|\bigr)}$.
\end{theorem}
A slightly weaker version of Theorem~\ref{thm:zeroerror} appears in~\cite{kane2017active} (see Theorem 4.1 there).
The next step is to de-randomize Theorem~\ref{thm:zeroerror} and obtain a deterministic comparison decision tree.

\begin{theorem}
\label{thm:det}
Let $H \subset \R^n$ be a finite set with inference dimension $d$. Then there exists a comparison decision tree which computes $\A{H}$,
whose query complexity is ${O((d + n\log (nd)) \log |H|)}$.
\end{theorem}
The proof of Theorem~\ref{thm:det} uses a \emph{double-sampling argument},
a technique originated in the study of uniform convergence bounds in statistical learning theory~\cite{vapnik1971uniform}.
The following corollary summarizes the above theorems concisely.
For $h \in \Z^n$ define $\|h\|_{\infty} = \max |h_i|$.

\begin{corollary}
\label{cor:general}
Let $H \subset \Z^n$ be such that $\|h\|_{\infty} \le w$ for all $h \in H$.
Then there exists a comparison decision tree computing $\A{H}$ whose query complexity is
${O\bigl(n\log (nw) \log |H|\bigr)}$.
\end{corollary}

\begin{proof}
Observe that $\|h\|_{1} \le n |h\|_{\infty} \le nw$. By Theorem~\ref{thm:infdim}, the inference dimension of $H$ is $d=O(n \log(nw))$.
The corollary now follows from Theorem~\ref{thm:det}.
\end{proof}

One can now verify that
Theorem~\ref{thm:ksum}, Theorem~\ref{thm:sortAB} and Theorem~\ref{thm:subsetsum} follow from
Corollary~\ref{cor:general} by setting $w=1$.

\paragraph{Paper organization.}
We begin with some preliminaries in Section~\ref{sec:prelim}. We prove Theorem~\ref{thm:infdim}
in Section~\ref{sec:infdim}. We prove Theorem~\ref{thm:zeroerror} in Section~\ref{sec:zeroerror}.
We prove Theorem~\ref{thm:det} in Section~\ref{sec:det}. We discuss further research and open
problems in Section~\ref{sec:open}.

\paragraph{An acknowledgement.}
We thank the Simons institute at Berkeley, where this work was performed, for their hospitality.

\section{Preliminaries}
\label{sec:prelim}

Let $H\subseteq\R^n$ be a finite set.
For every $x\in\R^n$, $\A{H}(x)$ denotes the function
\[\A{H}(x) \eqdef \bigl(\sign(\ip{x,h}): h \in H\bigr) \in \{-,0,+\}^H,\]
where $\sign:\R \to \{-,0,+\}$ is the sign function and $\ip{\cdot,\cdot}$ is the standard inner product in $\R^n$.
The following lemma is a variant of standard bounds
on the number of cells in a hyperplane arrangement.

\begin{lemma}\label{lem:cellcount}
Let $H \subset \R^n$ be a set of size $|H|=m$. Then
$\lvert\{\A{H}(x) : x\in\R^n\}\bigr\rvert \leq (2em)^n$.
\end{lemma}

\begin{proof}
It is well known that a set of $m$ hyperplanes partitions $\R^n$ to at most ${m \choose \le n}$ open cells.
The lemma follows by first choosing $i \le n$ linearly independent hyperplanes to which $x$ belongs, and then applying the above
bound to the remaining ones (restricted to a subspace of dimension $n-i$). Thus
\begin{align*}
\bigl\lvert\{\A{H}(x) : x\in\R^n\}\bigr\rvert 
&\leq 
\sum_{i=0}^{n}\binom{m}{i}\binom{m-i}{\leq n-i}
= 
\sum_{i=0}^n \sum_{j=0}^{n-i} {m \choose i} {m-i \choose j} \\
&= 
\sum_{s=0}^n \sum_{i=0}^s {m \choose s} {s \choose i} 
= 
\sum_{s=0}^n {m \choose s} 2^s 
\leq
{m \choose \le n} 2^n
\le (2em)^n,
\end{align*}
where the second equality follows from the identity ${m \choose i} {m-i \choose j} = {m \choose s} {s \choose i}$,
where $s=i+j$, and the last inequality follows from the well known upper bound 
${m \choose \le n} \leq (em/n)^n \leq (em)^n$.
\end{proof}

\section{Bounding the inference dimension}
\label{sec:infdim}

We prove Theorem~\ref{thm:infdim} in this section.

\restate{Theorem~\ref{thm:infdim}}{
The inference dimension of $H = \{ h\in \Z^n :\|h\|_1 \le w\}$ is $d=O(n \log w)$.
}

Let $S \subset \Z^n$ be such that $\|h\|_1 \le w$ for all $h \in S$.
We assume $|S|=d$ where $d$ is large enough to be determined later.
Fix $x \in \R^n$. We will show that there exists $h \in S$ such that $S \setminus \{h\}$ infers $h$ at $x$.

Partition $S$ into $\bigl\{S_b: b \in \{-,0,+\}\bigr\}$, where
\[
S_b \eqdef \{h \in S: \sign(\ip{h,x})=b\}.
\]

We will show that if $S$ is sufficiently large then
$S_b \setminus \{h\}$ infers $h$ at $x$ for some $s\in S_b$ and $b\in\{-,0,+\}$.
The simplest case is when $S_0$ is large:
\begin{claim}
If $|S_0|>n$ then there exists $h \in S_0$ such that $S_0 \setminus \{h\}$ infers $h$ at $x$.
In particular, $S \setminus \{h\}$ infers $h$ at $x$.
\end{claim}

\begin{proof}
Let $h_1,\ldots,h_{n+1} \in S_0$ be distinct elements such that $h_{n+1}$ belongs to the linear span of $h_1,\ldots,h_n$.
We claim that $\{h_1,\ldots,h_n\}$ infer $h_{n+1}$ at $x$.
More specifically, we claim that having
\begin{itemize}
\item[(i)] $\sign(\ip{h_{i},x})=0$ for $i\leq n$, and
\item[(ii)] $h_{n+1}\in\text{span}\{h_i:i\leq n\}$
\end{itemize}
imply that $\sign(\ip{h_{n+1},x})=0$.
Indeed, by (ii) there exist coefficients $\alpha_i$'s such that
$h_{n+1}=\sum_{i=1}^n\alpha_i h_i$,
and therefore, using (i), it follows that
$\ip{h_{n+1},x}= \ip{\sum_{i=1}^n\alpha_i h_i,x}=\sum_{i=1}^n \alpha_i\ip{h_i,x}=0$.
\end{proof}

Thus, we assume from now on that $|S_0| \le n$. We assume without loss of generality that $|S_{+}| \ge |S_{-}|$,
and show that there is some $h\in S_+$ such that $S_+\setminus\{h\}$ infers $h$ at $x$.
The other case is analogous.
Set $m=\lfloor (d-n)/2 \rfloor$ and let $h_1,\ldots,h_m \in S_+$ sorted by
\[
0 < \ip{h_1,x} \le \ldots \le \ip{h_m,x}.
\]
The idea is to show that some $h_i$ satisfies that $h_i-h_1$ is in the cone
spanned by the $h_k-h_l$ where $1\leq l \leq k < i$.
Then, a simple argument shows that $S_+\setminus\{h_i\}$ infers $h_i$ at $x$.
The existence of such an $h_i$ is derived by a counting argument
that boils down to the following lemma.

\begin{claim}
\label{claim:pigeonhole}
Assume that $2^{m-1}>(\tfrac{2e(2w+1)m}{n})^{n}$. Then there exist $\alpha_1,\ldots,\alpha_{m-1} \in \{-1,0,1\}$, not all zero, such that
$$
\sum_{i=1}^{m-1} \alpha_i (h_{i+1}-h_{i}) = 0.
$$
In particular, this holds for $m=O(n \log w)$ with a large enough constant.
\end{claim}

\begin{proof}
For any $\beta \in \{0,1\}^{m-1}$ define $f(\beta) \eqdef \sum \beta_i (h_{i+1}-h_{i})$.
Note that $f(\beta) \in \Z^n$, and as since $\|h_i\|_1 \le w$ for all $i$,
it follows that $\|f(\beta)\|_1 \le 2w(m-1)$ by the triangle inequality.
Let $F \eqdef \{f(\beta): \beta \in \{0,1\}^{m-1}\}$.
Next, we bound $|F|$. We claim that
\[
|F| \le 2^n {2w(m-1)+n \choose n}.
\]
To see that, note that there are $2^n$ possible signs for each $f \in F$. The number of patterns for the absolute values
is at most the number of ways to express $2w(m-1)$ as the sum of $n+1$ nonnegative integers. Equivalently, it is the number
of ways of placing $2w(m-1)$ balls in $n+1$ bins, which is ${2w(m-1)+n \choose n}$. We further simplify
$$
|F| \le 2^n {2w(m-1)+n \choose n} \le 2^n {(2w+1) m \choose n} \le \left(\frac{2e(2w+1)m}{n} \right)^n.
$$
By our assumptions $2^{m-1} > |F|$. Thus by the pigeonhole principle there exist distinct $\beta',\beta''$ for which $f(\beta')=f(\beta'')$.
The claim follows for $\alpha=\beta' - \beta''$.
\end{proof}

We assume that $d=O(n \log w)$ with a large enough constant, so that the conditions of Claim~\ref{claim:pigeonhole} hold.
Let $\alpha_1,\ldots,\alpha_{m-1} \in \{-1,0,1\}$, not all zero, be such that $\sum \alpha_i (h_{i+1}-h_{i})=0$.
Let $1 \le p \le m-1$ be maximal such that $\alpha_p \ne 0$.
We may assume that $\alpha_p=-1$, as otherwise we can negate all of $\alpha_1,\ldots,\alpha_{m-1}$.

Adding $h_{p+1} - h_1 = \sum_{i=1}^{p} (h_{i+1} - h_{i})$ to $0=\sum \alpha_i (h_{i+1}-h_{i})$, we obtain that
$$
h_{p+1} - h_1 = \sum_{i=1}^{p} (\alpha_i+1)(h_{i+1} - h_{i}) = \sum_{i=1}^{p-1} (\alpha_i+1)(h_{i+1} - h_{i}),
$$
where the first equality holds as $\alpha_i=0$ if $i>p$, and the second equality holds as $\alpha_{p}=-1$.

We claim that $R=\{h_1,\ldots,h_p\}$ infers $h_{p+1}$ at $x$, which completes the proof.
More specifically, we claim that having
\begin{itemize}
\item[(i)] $0 < \ip{h_1,x} \le \ldots \le \ip{h_p,x}$,
\item[(ii)] $h_{p+1} - h_1  = \sum_{i=1}^{p-1} (\alpha_i+1)(h_{i+1} - h_{i})$,
where the coefficients $\alpha_{i}+1\geq 0$ for all $i$,
\end{itemize}
imply that $\sign(\ip{h_{p+1},x})\geq 0$.
Indeed, item (i) implies that $\ip{x,h_i-h_j}\geq 0$,
for every $1\leq j < i \leq p$,
and item (ii) implies that $h_{p+1}-h_1$
is in the cone spanned by $h_i - h_j$ for $1\leq j < i \leq p$.
Thus, also $\ip{x,h_{p+1}-h_1}\geq 0$, which implies,
by the left-most inequality of item (ii),
that $\ip{x,h_{p+1}}\geq\ip{x,h_{1}}> 0$,
as required.


\section{Zero-error randomized comparison decision tree}
\label{sec:zeroerror}

We prove Theorem~\ref{thm:zeroerror} in this section.

\restate{Theorem~\ref{thm:zeroerror}}{
Let $H \subset \R^n$ be a finite set with inference dimension $d$. Then there exists a zero-error randomized comparison decision tree which computes $\A{H}$,
whose expected query complexity is ${O\bigl((d +n\log d) \log |H|\bigr)}$.
}

We begin with the following claim. Recall that $\infer(S,x)$ is the set of $h \in \R^n$ which can be inferred from $S$ at $x$.

\begin{claim}
\label{claim:inf_many}
Let $S \subset \R^n$ with inference dimension $d$ and $|S|=d+m$.
Then for every $x \in \R^n$, there exist $h_1,\ldots,h_m \in S$ such that
$$
h_i \in \infer(S \setminus \{h_i\}, x).
$$
\end{claim}

\begin{proof}
We apply the definition of inference dimension iteratively. Fix $x \in \R^n$.
Assume that we constructed $h_1,\ldots,h_{i-1}$ so far for $i \le m$. Let
$S_i=S \setminus \{h_1,\ldots,h_{i-1}\}$. As $|S_i| \ge d$ there exist $h_i \in S_i$ such that $S_i \setminus \{h_i\}$
infers $h_{i}$ at $x$. That is, $h_{i} \in \infer(S_{i} \setminus \{h_i\},x)$. But as $S_{i} \subset S$ then also
$h_{i} \in \infer(S \setminus \{h_i\},x)$.
\end{proof}

\begin{lemma}
\label{lemma:inf_half}
Let $H \subset \R^n$ be a finite set with inference dimension $d$.
Let $S \subset H$ be a uniformly chosen subset of size $|S|=2d$. Then for every $x \in \R^n$,
\[
\E_S\bigl[|\infer(S,x) \cap H| \bigr] \ge \frac{|H|}{2}.
\]
\end{lemma}

\begin{proof}
Fix $x \in \R^n$. We have
\begin{align*}
\E_S\left[\frac{|\infer(S,x) \cap H|}{|H|} \right] &=
\Pr_{S \subset H, h \in H}[h \in \infer(S,x)] \\
&\ge \Pr_{S \subset H, h \in H \setminus S}[h \in \infer(S,x)] \\
&= \Pr[h_{2d+1} \in \infer(\{h_1,\ldots,h_{2d}\},x)],
\end{align*}
where $h_1,\ldots,h_{2d+1} \in H$ are uniformly chosen distinct elements. The inequality
``$\Pr_{S \subset H, h \in H}[h \in \infer(S,x)] \ge \Pr_{S \subset H, h \in H \setminus S}[h \in \infer(S,x)]$"
follows as $h \in \infer(S,x)$ for any $h \in S$.

Let $R \eqdef \{h_1,\ldots,h_{2d+1}\}$. By symmetry it holds that
\begin{align*}
\Pr[h_{2d+1} \in \infer(\{h_1,\ldots,h_{2d}\},x)] &= \frac{1}{2d+1} \sum_{i=1}^{2d+1} \Pr_R[h_i \in \infer(R \setminus \{h_i\},x)]\\
&= \E_R \left[\frac{|\{h_i \in R: h_i \in \infer(R \setminus \{h_i\},x)\}|}{2d+1} \right].
\end{align*}
By Claim~\ref{claim:inf_many}, for any $R \subset H$ it holds that $|\{h_i \in R: h_i \in \infer(R \setminus \{h_i\},x)\}| \ge |R|-d$.
Thus,
$$
\E_S\left[\frac{|\infer(S,x) \cap H|}{|H|} \right] \ge \frac{d+1}{2d+1} \ge \frac{1}{2}.
$$
\end{proof}

We are now in position to describe the algorithm which establishes Theorem~\ref{thm:zeroerror}.

\begin{tcolorbox}
\begin{center}
{\bf Zero-error randomized comparison decision tree for $\A{H}$}\\
\end{center}
\noindent
Input: $x \in \R^n$\\
Output: $\A{H}(x)$
\begin{enumerate}
\item[(1)] Initialize: $H_0=H$, $i=0$, $v(h)=?$ for all $h \in H$.
\item[(2)] Repeat while $|H_i| \ge 2d$:
\begin{enumerate}
\item[(2.1)] Sample uniformly $S_i \subset H_i$ of size $|S_i|=2d$.
\item[(2.2)] Query $\sign(\ip{h,x})$ for $h\in S_i$ and sort the $\ip{h,x}$ using comparison queries.
\item[(2.3)] Compute $\infer(S_i,x) \cap H_i$.
\item[(2.4)] For all $h \in \infer(S_i,x) \cap H_i$, set $v(h) \in \{-,0,+\}$ to be the inferred value of $h$ at $x$.
\item[(2.5)] Set $H_{i+1} := H_i \setminus (\infer(S_i,x) \cap H_i)$.
\item[(2.6)] Set $i := i+1$.
\end{enumerate}
\item[(3)] Query $\sign(\ip{h,x})$ for all $h \in H_i$, and set $v(h)$ accordingly.
\item[(4)] Return $v$ as the value of $\A{H}(x)$.
\end{enumerate}
\end{tcolorbox}

\paragraph{Analysis.}
In order to establish Theorem~\ref{thm:zeroerror}, we first show that for every $x \in \R^n$, the algorithm
terminates after $O(\log |H|)$ iterations in expectation. This follows as $\E[|H_i|] \le 2^{-i} |H|$,
which we show by induction on $i$. It clearly holds for $i=0$. For $i>0$
by Lemma~\ref{lemma:inf_half}, if we condition on $H_{i-1}$ then
$$
\E_{S_i}[|H_{i}| \; | \; H_{i-1}] \le \frac{|H_{i-1}|}{2}.
$$
and hence
$$
\E[|H_{i}|] = \E_{H_{i-1}}[\E_{S_{i}}[|H_{i}| \; | \; H_{i-1}]] \le \E\left[\frac{|H_{i-1}|}{2}\right] \le 2^{-i} |H|.
$$
Thus, it remains to bound the number of queries in every round.
Observe that the only queries to $x$ are in steps (2.2) and (3). In step (3) the algorithm makes at most $2d$
label queries. In step (2.2), we need to compute $\sign(\ip{x,h})$ for all $h \in S_i$,
which requires $|S_i|=2d$ label queries; and to compute $\sign(\ip{x,h'-h''})$ for all $h',h'' \in S_i$.
This can be done in $O(d \log d)$ comparison queries by sorting the elements $\{\ip{x,h}: h \in S_i\}$
giving some $O(d\log d\log|H|)$  bound on the expected total number of queries.

This bound can be improved using Fredman's sorting algorithm~\cite{fredman1976good}.

\begin{theorem}[\cite{fredman1976good}]
Let $\Pi$ be a family of orderings over a set of $m$ elements.
Then there exists a comparison decision tree
that sorts every $\pi\in\Pi$ using at most
\[2m + \log\lvert\Pi\rvert\]
comparisons.
\end{theorem}
To use Fredman's algorithm, observe that the ordering, ``$\prec$'',
on $S_i$ that is being sorted in the $i$'th round is defined by the inner product with $x$,
\[h' \prec h'' \iff \ip{h',x}\leq\ip{h'',x}.\]
The following claim bounds the number of such orderings.

\begin{claim}
Let $S \subset \R^n$. Let $\Pi_{S,x}$ be the ordering on $S$ define by inner product with $x \in \R^n$. Then
\[|\{\Pi_{S,x} : x\in\R^n\}| \leq (2e |S|^2)^{n}.\]
\end{claim}

\begin{proof}
Observe that $\Pi_{S,x'}\neq\Pi_{S,x''}$  if and only if
there are $h',h''\in S$ such that $\sign(\ip{h'-h'',x'})\neq\sign(\ip{h'-h'',x''})$.
Thus, the number of different orderings is at most the size of
$\{\A{S-S}(x) : x\in\R^n\}$, where $S-S=\{h'-h'' : h',h''\in S\}$.
Since $|S-S|\leq |S|^2$, Lemma~\ref{lem:cellcount} implies an upper bound of $(2e |S|^2)^{n}$
as claimed.
\end{proof}

Thus, by using Fredman's algorithm we can sort $S_i$ with just
${O(|S_i| + n\log|S_i|)}={O(d+n \log d)}$ comparisons in each round,
which gives a total number of
\[
O((d+n\log d)\log|H|)
\]
queries in total.

\section{Deterministic comparison decision tree}
\label{sec:det}

We prove Theorem~\ref{thm:det} in this section, which is a de-randomization of Theorem~\ref{thm:zeroerror}.

\restate{Theorem~\ref{thm:det}}{
Let $H \subset \R^n$ with inference dimension $d$. Then there exists a deterministic comparison decision tree which computes $\A{H}$,
whose query complexity is ${O((d + n\log (nd)) \log |H|)}$.
}

First, note the following straightforward Corollary of Lemma~\ref{lemma:inf_half}.

\begin{corollary}
\label{cor:inf_half}
Let $H \subset \R^n$ be a finite set with inference dimension $d$. Let $S \subset H$ be uniformly chosen
of size $|S|=2d$. Then
\[
\bigl(\forall x \in \R^n\bigr): \; \Pr_S \left[ |\infer(S,x) \cap H| \ge \frac{|H|}{4} \right] \ge \frac{1}{4}.
\]
\end{corollary}

Theorem~\ref{thm:det} follows by establishing a universal set $S$ which is good for all $x \in \R^n$.

\begin{lemma}
\label{lemma:inf_half_uni}
Let $H \subset \R^n$ be a finite set with inference dimension $d$.
Then there exists $S\subseteq H$ of size $|S|=O(d + n\log d)$ such that:
\[
\bigl(\forall x \in \R^n\bigr): \; |\infer(S,x) \cap H| \ge \frac{|H|}{8}.
\]
\end{lemma}

We first argue that Theorem~\ref{thm:det} follows directly from the existence of such an $S$.
The algorithm is a straightforward
adaptation of the zero-error randomized comparison algorithm,
except that now we use this set $S$ which works for all $x \in \R^n$ in parallel.

\begin{tcolorbox}
\begin{center}
{\bf Deterministic comparison decision tree for $\A{H}$}\\
\end{center}
\noindent
Input: $x \in \R^n$\\
Output: $\A{H}(x)$
\begin{enumerate}
\item[(1)] Initialize: $H_0=H$, $i=0$, $v(h)=?$ for all $h \in H$. Let $s=O(d+n \log d)$ as in Lemma~\ref{lemma:inf_half_uni}.
\item[(2)] Repeat while $|H_i| \ge s$:
\begin{enumerate}
\item[(2.1)] Pick $S_i \subset H_i$ of size $|S_i| = s$ such that
\[
\forall x \in \R^n, \; |\infer(S_i,x) \cap H| \ge \frac{|H|}{8}.
\]
\item[(2.2)] Query $\sign(\ip{h,x})$ for $h\in S_i$ and sort the $\ip{h,x}$ using comparison queries.
\item[(2.3)] Compute $\infer(S_i,x) \cap H_i$.
\item[(2.4)] For all $h \in \infer(S_i,x) \cap H_i$, set $v(h) \in \{-,0,+\}$ to be the inferred value of $h$ at $x$.
\item[(2.5)] Set $H_{i+1} := H_i \setminus (\infer(S_i,x) \cap H_i)$.
\item[(2.6)] Set $i := i+1$.
\end{enumerate}
\item[(3)] Query $\sign(\ip{h,x})$ for all $h \in H_i$, and set $v(h)$ accordingly.
\item[(4)] Return $v$ as the value of $\A{H}(x)$.
\end{enumerate}
\end{tcolorbox}

\paragraph{Analysis.}
Lemma~\ref{lemma:inf_half_uni} ensures that a set $S_i$ always exist.
Thus, for any $x$, the algorithm terminates after $O(\log |H|)$ rounds.
Observe that the only queries to $x$ are in steps (2.2) and (3).
In step (3) the algorithm makes at most $s = O(d + n \log d)$ label queries. In step (2.2), we need to compute
$\sign(\ip{x,h})$ for all $h \in S_i$,
and to compute $\sign(\ip{x,h'-h''})$ for all $h',h'' \in S_i$,
which can be done sorting the elements $\{\ip{x,h}: h \in S_i\}$.
Using Fredman's algorithm, this requires
$O(|S_i| + n\log|S_i|)=O(d+n \log(dn))$ many comparisons in each round,
which gives a total number of
\[
O((d + n\log (dn)) \log |H|)
\]
queries.
\subsection{Proof of Lemma~\ref{lemma:inf_half_uni}}
Let $S \subset H$ be a uniform subset of size $|S|=s$ where $s=O(d+n \log d)$.
Define the event
\[
E(S) \eqdef \left[ \exists x \in \R^n, |\infer(S,x) \cap H| < \frac{|H|}{8} \right].
\]
It suffices to prove that $\Pr[E(S)] < 1$ to prove the existence of $S$.
In fact, as we will see, by choosing sufficiently large constants in the choice of $s=O(d+n\log d)$,
the probability $\Pr[E(S)]$ can be made $\le 1/2$ (say), so a random set would also work.

In order to establish that  $E(S) < 1$ we use a variant of
the \emph{double sampling method}~\cite{vapnik1971uniform} (see also~\cite{vapnik2015uniform}).
Let $T \subset S$ be a uniformly chosen subset of size $|T|=2d$.
Define the event
\[
E(S,T) \eqdef \left[ \exists x \in \R^n,\;
|\infer(T,x) \cap H| < \frac{|H|}{8} \; \bigwedge \;
|\infer(T,x) \cap S| \ge \frac{|S|}{4} \right].
\]
We bound  $\Pr(E(S))$ in two steps.
We first show that
(i) $\Pr[E(S)] \le 4 \Pr[E(S,T)]$,
and then that (ii) $\Pr[E(S,T)] \le \frac{1}{8}$.

\begin{claim}
\label{claim:ST}
$\Pr[E(S)] \le 4 \Pr[E(S,T)]$.
\end{claim}

\begin{proof}
For each $S$ for which $E(S)$ holds fix $x_S \in \R^n$ such that $|\infer(S,x_S) \cap H| < \frac{|H|}{8}$.
Then
$$
\Pr[E(S,T) \;|\; S] \ge \Pr \left[ |\infer(T,x_S) \cap H| < \frac{|H|}{8} \; \bigwedge \;
|\infer(T,x_S) \cap S| \ge \frac{|S|}{4} \right].
$$
The first condition holds with probability one, since $T \subset S$ and hence $\infer(T,x_S) \subset \infer(S,x_S)$.
For the second condition, as $T \subset S$ is a uniformly chosen subset of size $|T|=2d$, Corollary~\ref{cor:inf_half} gives
$$
\Pr_T\left[|\infer(T,x_S) \cap S| \ge \frac{|S|}{4} \; \bigg| \; S\right] \ge \frac{1}{4}.
$$
Thus
$$
\Pr[E(S,T) \;|\; S] \ge \frac{1}{4}
$$
As this holds for every $S$ for which $E(S)$ holds, we have $\Pr[E(S,T)|E(S)] \ge 1/4$, which implies the claim.
\end{proof}

We next bound the probability of $E(S,T)$. We will prove that for every fixed $T$,
$$
\Pr[E(S,T) \; | \; T] \le \frac{1}{8},
$$
which will conclude the proof.
So, fix $T \subset H$ of size $|T|=2d$.
Let $T-T$ denote the set $\{h'-h'' : h',h''\in T\}$,
and let $T^*= T\cup (T-T)$.
Recall that $\A{T^*}(x)$
is defined by
\[
\A{T^*}(x)= (\sign(\ip{h,x}): h \in T^*) \in \{-,0,+\}^{T^*}.
\]
Observe that the set $\infer(T,x)$ depends only on $\A{T^*}(x)$;
that is, if $\A{T^*}(x')=\A{T^*}(x'')$ then $\infer(T,x')=\infer(T,x'')$.
Let $X_T \subset \R^n$ be a set that contains one representative from each equivalence class of the relation
$x'\sim x'' \iff \A{T^*}(x') = \A{T^*}(x'')$.
Thus we can rephrase the event $E(S,T)$ as
\[
E(S,T) = \left[ \exists x \in X_T,\;
|\infer(T,x) \cap H| < \frac{|H|}{8} \; \bigwedge \;
|\infer(T,x) \cap S| \ge \frac{|S|}{4} \right].
\]
The advantage of considering $X_T$ is that now
we can bound the probability of $E(S,T)$ using
a union bound that depends on the (finite) set $X_T$.
More specifically, let
$$
X'_T \eqdef \left\{x \in X_T:\lvert\infer(T,x) \cap H\rvert < \frac{|H|}{8} \right\}.
$$
We thus established the following claim.
\begin{claim}
For every $T \subset H$,
$$
\Pr[E(S,T) \; | \; T] \le \sum_{x \in X'_T} \Pr_S \left[|\infer(T,x) \cap S| \ge \frac{|S|}{4} \; \bigg| \; T\right].
$$
\end{claim}

To conclude, it suffices to upper bound $|X'_T|$ and the probability that
$|\infer(T,x) \cap S| \ge \frac{|S|}{4}$ for $x \in X'_T$.
Lemma~\ref{lem:cellcount} gives an upper bound on $|X_T|$ which also bounds $|X'_T|$,
\[\lvert X'_T\rvert \leq \lvert X_T\rvert = |\A{T^*}| \le (2e |T^*|)^n = 2^{O(n \log d)}.\]
%
%
%
We next bound the probability (over $S\supset T$) that
$|\infer(T,x) \cap S| \ge \frac{|S|}{4}$ for $x \in X'_T$.
\begin{claim}
Fix $T \subset H$ of size $|T|=2d$ and fix $x \in X'_T$. Assume that $s \ge 10 |T|$, and let $S$ be a uniformly sampled set of size $|S|=s$
such that $T \subset S\subset H$. Then
$$
\Pr_T\left[|\infer(T,x) \cap S| \ge \frac{|S|}{4} \;\Big\vert \; T \right] \le 2^{-\Omega(s)}.
$$
\end{claim}

\begin{proof}
Let $R=S \setminus T$.
It suffices to bound the probability of the event that
$|\infer(T,x) \cap R| \ge \frac{|R|}{6}$.
Indeed, if $|\infer(T,x) \cap S| \ge \frac{|S|}{4}$ then
\[
|\infer(T,x) \cap R| \ge \frac{|S|}{4}-|T|  =  \frac{|R|+|T|}{4}-|T|\geq \frac{|R|}{6},
\]
where in the last inequality we used the assumption that $|R|\geq 9|T|$.

The set $R$ is a uniform subset of $H \setminus T$ of size $|R|=|S|-|T|$.
By assumption, at most $\frac{|H \setminus  T|}{8}$ of the elements
in $H \setminus T$ are in $\infer(T,x)$.
By the Chernoff bound,
the probability that at least $|R|/6$ of the sampled elements belong to $\infer(T,x)$
is thus exponentially small in $|R|$. This finishes the proof as $|R| \ge (9/10)s$.
\end{proof}

We now conclude the proof.
\[
\Pr[E(S,T) \; | \; T] \le |X'_T|  2^{-\Omega(s)} \le 2^{O(n \log d) - \Omega(s)} \le 1/8,
\]
as we choose $s=O(d+n \log d)$ with a large enough hidden constant. Then we also have $\Pr[E(S,T)] \le 1/8$ and
$$
\Pr[E(S)] \le 4 \Pr[E(S,T)] \le 1/2.
$$

%
%
%

\section{Further research}
\label{sec:open}

We prove that many combinatorial point-location problems have near optimal linear decision trees. Moreover,
these are comparison decision trees,
in which the linear queries are particularly simple: both sparse (in many cases) and have only $\{-1,0,1\}$ coefficients.
This raises the possibility of having improved algorithms for these problems in other models of computations.
To be concrete, we focus on $3$-SUM below, but the same questions can be asked for any other problem of a similar flavor.

\paragraph{Uniform computation.}
The most obvious question is whether the existence of a near optimal linear decision tree implies anything about uniform computation.
As showed in~\cite{gronlund2014threesomes}, this can lead to log-factor savings. It is very interesting whether greater savings can be achieved.
We do not discuss this further here, as this question has been extensively discussed in the literature (see e.g.~\cite{vassilevska2015hardness}).

\paragraph{Nonuniform computation.}
Let $A \subset \R$ be a set of size $|A|=n$. It is very easy to ``prove'' that $A$ is a positive instance of $3$-SUM,
by demonstrating three elements whose sum is zero.
However, it is much less obvious how to prove that $A$ is a negative instance of $3$-SUM.
This problem was explicitly studied in~\cite{carmosino2016nondeterministic}
in the context of nondeterministic ETH. They constructed such a proof which can be verified in time $O(n^{3/2})$. It seems plausible that
our current approach may lead to improved bounds. Thus, we propose the following problem.

\begin{open}
Given a set of $n$ real numbers no three of which sums to 0.
Is there a proof of that fact which can be verified in near-linear time?
\end{open}

\paragraph{$3$-SUM with preprocessing.}
Let $A \subset \R$ of size $|A|=n$. The $3$-SUM with preprocessing problem allows one to preprocess the set $A$ in quadratic time.
Then, given any subset $A' \subset A$, the goal is to solve that $3$-SUM problem on $A'$ in time significantly faster then $n^2$.
Chan and Lewenstein~\cite{chan2015clustered} designed such an algorithm, which solves that $3$-SUM problem on any subset in time $O(n^{2-\eps})$
for some small constant $\eps>0$. It is interesting whether our techniques can help improve this to near-linear time.

\begin{open}
Given a set of $n$ real numbers, can they be preprocessed in $O(n^2)$ time, such that later on, for every subset of the numbers
the $3$-SUM problem can be solved in time near-linear in $n$?
\end{open}

\paragraph{General point-location problem.}
It is natural to ask whether the techniques used in this paper, and in particular, the inference-dimension,
can be used to improve the state-of-the-art upper bounds for general point location problems.
Unfortunately, unless the set of hyperplanes $H$ has some combinatorial structure,
its inference dimension may be unbounded: in~\cite{kane2017active} we construct examples of $H \subset \R^3$ whose inference dimension is unbounded.
Nevertheless, we conjecture that by generalizing comparison queries
(which are $\pm 1$ linear combinations of two elements in $H$)
to arbitrary linear combinations of two elements from $H$ might solve the problem.
\begin{conjecture}
Let $H \subset \R^n$. There exists a linear decision tree which computes $\A{H}$ of depth $O(n \log |H|)$. Moreover,
all the linear queries are in $\{\alpha h' + \beta h'': \alpha,\beta \in \R, h',h'' \in H\}$.
\end{conjecture}

\paragraph{Optimal bounds.}
We suspect that our analysis can be sharpened to improve the log-factors that separate it from the information theoretical lower bounds.
For concreteness, we pose the following conjecture.
\begin{conjecture}
For any $H \subset \{-1,0,1\}^n$ there exists a comparison decision tree which computes $\A{H}$ with $O(n \log |H|)$ many queries.
In particular,
\begin{itemize}
\item $3$-SUM on $n$ real numbers can be solved by a $6$-sparse linear decision tree which makes $O(n \log n)$ queries.
\item Sorting $A+B$, where $A,B$ are sets of $n$ real numbers, can be solved
by a $4$-sparse linear decision tree which makes $O(n \log n)$ queries.
\item SUBSET-SUM on $n$ real numbers can be solved
by a linear decision tree which makes $O(n^2)$ queries.

\end{itemize}
\end{conjecture}
Note that Corollary~\ref{cor:general} gives a bound of $O(n \log{n} \log|H|)$ for this problem. So, the goal is
to shave the $\log n$ factor.

\bibliographystyle{alpha}
\bibliography{LDT}

\end{document}